\documentclass[12pt,a4paper]{ouparticle}

\usepackage{apalike}

\usepackage{amsmath}
\usepackage{amsfonts}
\usepackage{dsfont}
\usepackage{amssymb}
\usepackage{amsthm}
\usepackage[dvipsnames]{xcolor}

\usepackage[colorlinks,citecolor=Blue,urlcolor=LimeGreen,bookmarks=false,hypertexnames=true]{hyperref}
\usepackage{cite}

\begin{document}

\newtheorem{thm}{\textsl{Theorem}}[section]
\newtheorem{dfn}[thm]{\textsl{Definition}}
\newtheorem{prp}[thm]{\textsl{Proposition}}
\newtheorem{cor}[thm]{\textsl{Corollary}}
\newtheorem{lem}[thm]{\textsl{Lemma}}
\newtheorem{rmk}[thm]{\textsl{Remark}}
\newtheorem{ex}[thm]{\textsl{Example}}

\numberwithin{equation}{section}

\title{\textbf{Spectral risk measures and uncertainty}}

\author{%
\name{Mohammed Berkhouch\footnote{\textit{Corresponding author. \url{https://orcid.org/0000-0001-7207-4517}}.}}
\address{\textit{LISAD, ENSA, Ibn Zohr University, Agadir, Morocco.}}
\email{\texttt{E-mail: mohammed.berkhouch@edu.uiz.ac.ma}}
\and
\name{Ghizlane Lakhnati}
\address{\textit{LISAD, ENSA, Ibn Zohr University, Agadir, Morocco.}}
\email{\texttt{E-mail: g.lakhnati@uiz.ac.ma}}
\and
\name{Marcelo Brutti Righi}
\address{\textit{Federal University of Rio Grande do Sul, Porto Alegre, Brazil.}}
\email{\texttt{E-mail: marcelo.righi@ufrgs.br}}}

\abstract{ Risk assessment under different possible scenarios is a source of uncertainty that may lead to concerning financial losses. We address this issue, first, by adapting a robust framework to the class of spectral risk measures. Second, we propose a Deviation-based approach to quantify uncertainty. Furthermore, the theory is illustrated with a practical application.
\begin{flushleft}
\end{flushleft}
}

\date{\today}


\maketitle

\section{Introduction}
\label{sec1}

In risk management, the ultimate goal is to calculate required capital to act as a buffer against the inherent risk of a financial position. Over the last few decades, numerous risk measures, which are mappings from a set of random variables to real numbers, have been introduced. Typical examples are Value-at-Risk, Expected Shortfall and various coherent and convex risk measures introduced, respectively, by \cite{artzner1999coherent} and  \cite{follmer2002convex} as axiomatic approach of the reasonable theoretical properties that a measure of risk may fulfill (see Definition \ref{a}). For a comprehensive review, we recommend \cite{pflug2007modeling}, \cite{delbaen2012monetary}, \cite{alexander2005quantitative} and  \cite{föllmer2016stochastic}.

Among coherent risk measures the only law-invariant and co-monotonic additive ones are spectral risk measures, introduced in \cite {acerbi2002spectral} and arguably considered as the most important extensions of Expected Shortfall. Beyond fulfilling most of the suitable theoretical properties of a reasonable risk measure, a \textit{spectral risk measure} is characterized with a weighting function $\phi$ that account for the psychological attitude of different profiles, by reflecting risk-aversion. This argumentum makes this class of risk measures of interest and justifies the framework of the present paper.

 In practice, risk measures have to be estimated from data. Therefore, it is totally legitimate to argue that the most suitable risk measures, for use in practice, are the law-invariant (\textit{distribution-based}) ones. Under different \textit{scenarios} we get different assessments of risk. Moreover, improper distribution assumptions (\textit{probability measures}) can largely affect the risk value, lead to wrong decision and then to significant financial losses. Therefore, the choice of probability measures is a source of \textit{uncertainty} in the risk measurement process. This situation naturally requires considering \textit{robust} risk measures; thus, risk measures that are insensitive to the choice of probability measures.

That is to say, in risk measurement practitioners use risk measures based on a given probability measure on $(\Omega, \mathcal{F})$. Therefore, for each probability we generally get a different measure of risk. These probability measures could be understood as alternative \textit{scenarios}, that we can indifferently interpret as different (models, values of an estimated parameter, economic situations, beliefs...). However, this gives rise to the question: Is there a decent probability measure? Typically, we have a set of candidate probability measures and often we do not know how to pick the appropriate one due to the uncertain character of the financial market.

This \textit{uncertainty} regarding the choice of the proper probability measure (then the proper risk measure) has motivated the investigation of two issues:
\begin{description}
  \item[$\centerdot$] How to overcome uncertainty?
  \item[$\centerdot$] How to measure uncertainty?
\end{description}

In this paper, each probability measure is associated with a scenario. We intend to address the above formulated questions by adapting a consistent framework to the class of spectral risk measures. In this sense, our paper contributes to the existent literature recently addressing this issue, see \cite{wang2018scenario}, \cite{jokhadze2018measuring}, \cite{Righi2018theory} and the references therein.  Moreover, we intend to propose an alternative approach for measuring uncertainty itself. Since quantifying uncertainty provides us with information about how far our risk measurement process could be impacted by uncertainty; and it may even be seen as a penalization to add to the required capital for covering the held financial position. To measure uncertainty \cite{jokhadze2018measuring} proposed a superposed measure that evaluates the dispersion, of a collection of risk measure, relatively to some reference risk measure $\rho_{0}$. At this point, the approach suggested by the authors is interesting but it stays dependent on the choice of a risk measure of reference, and since we are dealing with uncertainty this approach does not completely address the issue. Moreover, the monetary risk measures used in \cite{jokhadze2018measuring}, such as $VaR$ and $ES$, do not capture the variability concept.

The rest of the paper is structured as follows: Section 2 exposes preliminaries regarding notations, theoretical review of risk and deviation measures and examples as well. In section 3, we study two approaches of a robust framework for spectral risk measures. To measure uncertainty an alternative approach is suggested in section 4. Section 5 illustrates our contribution in a case study.

\section{Preliminaries}
\label{sec2}

Consider a measurable space $(\Omega, \mathcal{F})$ and let $\mathcal{P}$ the set of all probability measures on $(\Omega, \mathcal{F})$. For a probability measure $\mathrm{P}\in \mathcal{P}$, $L^{\infty}(\Omega, \mathcal{F}, \mathrm{P})$ is the space of equivalent classes  of essentially bounded random variables. We denote $X$ the random future outcome of a financial position. Constant random variables are identified with real numbers. We use $\mathbb{E}^{\mathrm{P}}[X]= \int_{\Omega}X d\mathrm{P}$, $F_{X,\mathrm{P}}(x)= \mathrm{P}(X\leq x)$ and $F^{-1}_{X,\mathrm{P}}(\alpha)= \inf \{x: F_{X,\mathrm{P}}(x)\geq \alpha\}$ to denote, respectively, the expected value, the probability function (c.d.f.) and its generalized inverse for $X$ under $\mathrm{P}\in \mathcal{P}$. We say that a pair of random variables $X,Y \in L^{\infty}$ is co-monotone if $(X(\omega)-X(\omega'))(Y(\omega)-Y(\omega'))\geq 0 \text{ for all } (\omega,\omega')\in \Omega\times\Omega$.

We expose definitions and theoretical properties of both risk and deviation measures.

\begin{dfn}\label{a}
A risk measure is a functional $\rho: L^{\infty}\rightarrow\mathds{R}$, which may satisfy the following properties:
\begin{description}
  \item[$\bullet$] \textit{Law Invariance:} If $X\in L^{\infty}$ and $Y\in L^{\infty}$ have the same distribution under $\mathrm{P}$, succinctly $X\stackrel{d}{=_{\mathrm{P}}}Y$, then $\rho(X)=\rho(Y)$.
  \item[$\bullet$] \textit{Monotonicity:} $\rho(X)\leq \rho(Y)$ when $X,Y \in L^{\infty}$ are such that $X\geq Y$.
  \item[$\bullet$] \textit{Translation Invariance:} $\rho (X+C)=\rho(X)-C$ for all $C\in \mathds{R}$ and $X\in L^{\infty}$.
  \item[$\bullet$] \textit{Sub-Additivity:} $\rho(X+Y)\leq\rho(X)+\rho(Y)$ for every pair $X,Y\in L^{\infty}$.
  \item[$\bullet$] \textit{Positive Homogeneity:} $\rho (\lambda X)=\lambda \rho(X)$ for all $\lambda \geq 0$ and $X\in L^{\infty}$.
  \item[$\bullet$] \textit{Convexity:} $\rho(\lambda X+ (1-\lambda)Y)\leq \lambda\rho(X)+(1-\lambda)\rho(Y)$ for all $X,Y\in L^{\infty}$ and $\lambda \in [0,1]$.
  \item[$\bullet$] \textit{Co-monotonic Additivity:} $\rho(X+Y)=\rho(X)+\rho(Y)$ for every co-monotonic pair $X,Y\in L^{\infty}$.
  \item[$\bullet$] \textit{Fatou Continuity:} if $\displaystyle \lim_{n\rightarrow \infty} X_{n}=X$ and $\{X_{n}\}_{n=1}^{\infty}, X\in L^{\infty}$, then $\rho(X)\leq \displaystyle\liminf_{n\rightarrow \infty}\rho(X_{n})$.
\end{description}

A risk measure is called \textit{distribution-based} or \textit{law-invariant} if it satisfies Law Invariance, \textit{monetary} if it fulfills Monotonicity and Translation Invariance, \textit{convex} if it is monetary and respects Convexity, \textit{coherent} if it is convex and possesses Positive Homogeneity, \textit{co-monotone} if it attends Co-monotonic Additivity, and \textit{Fatou Continuous} if it satisfies Fatou Continuity.
\end{dfn}

\begin{rmk}
Every pair from Sub-Additivity, Positive Homogeneity and Convexity properties implies the third one. For financial interpretations of the properties above, we refer the reader to (\cite{follmer2002convex}, Chap 4), \cite{delbaen2012monetary} and \cite{alexander2005quantitative}.
\end{rmk}

\begin{ex}\label{b}
The functionals provided below, excepting $WC$, are examples of distribution-based risk measures:
\begin{description}
  \item[$\bullet$] \textit{Worst-case (WC):} This is an extreme robust risk measure, since it does not depend on the probability $\mathrm{P}$, defined as:
      \begin{equation}\label{1}
        WC(X)= -\inf X, \,\forall X\in L^{\infty}.
      \end{equation}
  \item[$\bullet$] \textit{Expected Loss (EL):} This is a parsimonious law-invariant co-monotone and coherent risk measure defined as:
   \begin{equation}\label{2}
    EL^{\mathrm{P}}(X)= -\mathbb{E}^{\mathrm{P}}[X]= -\int_{0}^{1}F_{X,\mathrm{P}}^{-1}(\gamma)d\gamma, \,\forall X\in L^{\infty}.
      \end{equation}

  \item[$\bullet$] \textit{Value-at-Risk (VaR):} This is a leading law-invariant monetary risk measure in both financial theory and practice, defined conform:
      \begin{equation}\label{3}
   VaR_{\alpha}^{\mathrm{P}}(X)= -\inf \{x: F_{X,\mathrm{P}}(x)\geq \alpha\}, \alpha \in [0,1], \,\forall X\in L^{\infty}.
      \end{equation}

  \item[$\bullet$] \textit{Entropic risk measure (Entr):} This is a law-invariant convex risk measure characterized with an aversion parameter $\tau >0$, defined as:
       \begin{equation}\label{4}
    Entr_{\tau}^{\mathrm{P}}(X)= \frac{1}{\tau}\log\mathbb{E}^{\mathrm{P}}[e^{-\tau X}]= \displaystyle \sup_{\mathrm{Q}}\left\{\mathbb{E}^{\mathrm{Q}}[-X]-\frac{1}{\tau}H(\mathrm{Q}/\mathrm{P})\right\}, \,\forall X\in L^{\infty}.
      \end{equation}

     where, $H$ denotes the relative entropy of ${\mathrm{Q}}$ with respect to ${\mathrm{P}}$, see \cite{follmer2002convex}.
  \item[$\bullet$] \textit{Expected Shortfall (ES):} This is a prominent law-invariant co-monotone and coherent risk measure defined as:
  \begin{equation}\label{5}
  ES^{\mathrm{P}}(X)= \frac{1}{1-\alpha}\int_{\alpha}^{1}VaR_{\gamma}^{\mathrm{P}}(X)d\gamma, \alpha \in [0,1),\,\forall X\in L^{\infty}.
      \end{equation}

  \item[$\bullet$] \textit{Spectral risk measure ($\rho_{\phi}$):} This is the class of law-invariant co-monotone and coherent risk measures, proposed by \cite{acerbi2002spectral}, expressed as:
      \begin{equation}\label{6}
   \rho_{\phi}^{\mathrm{P}}(X)=\int_{0}^{1}VaR_{\gamma}^{\mathrm{P}}(X)\,\phi(\gamma)d\gamma,\,\forall X\in L^{\infty}.
      \end{equation}

      where, $\phi$ is a non-increasing, non-negative, right continuous and integrable weighting function such that $\int_{0}^{1}\phi(\gamma)d\gamma=1$.
\end{description}
\end{ex}

\begin{dfn}\label{c}
A deviation measure is a functional $D: L^{\infty}\rightarrow\mathds{R_{+}}$, \cite{rockafellar2006generalized}, which may fulfills the following properties:

\begin{description}
  \item[$\bullet$] \textit{Non-Negativity:} For all $X\in L^{\infty}$, $D(X)=0$ for constant $X$ and $D(X)>0$ for non-constant $X$.
  \item[$\bullet$] \textit{Translation Insensitivity:} $ D(X+C)= D(X)$ for all $ C\in \mathds{R}$ and $X\in L^{\infty}$.
  \item[$\bullet$] \textit{Convexity:} $D(\lambda X+(1-\lambda)Y)\leq \lambda D(X)+(1-\lambda)D(Y)$ for all $X,Y\in L^{\infty}$.
  \item[$\bullet$] \textit{Positive Homogeneity:} $ D(\lambda X)=\lambda D(X)$ for all $\lambda \geq 0$ and $X\in L^{\infty}$.
  \item[$\bullet$] \textit{Co-monotonic additivity:} $D(X+Y)=D(X)+D(Y)$ for every co-monotonic pair $X,Y\in L^{\infty}$.
\end{description}

A measure of deviation $D$ is \textit{proper} if it satisfies Non-Negativity and Translation Insensitivity, \textit{convex} if it is proper and fulfills Convexity, \textit{coherent} if it is convex and respects Positive Homogeneity, and \textit{co-monotone} if it attends Co-monotonic additivity. See e.g., \cite{righi2018composition}, \cite{furman2017gini} and \cite{berkhouch2018extended}.
\end{dfn}

\begin{ex} We provide below some examples that illustrate the deviation concept:
\begin{description}
  \item[$\bullet$] \textit{Full Range (FR):} This extremely conservative deviation measure, that represents the larger possible difference between two values of $X(\omega)$ for distinct $\omega\in \Omega$, is defined as:
      \begin{equation}\label{7}
        FR(X)= \sup X -\inf X, \,\forall X\in L^{\infty}.
      \end{equation}

  \item[$\bullet$] \textit{Lower and Upper Range (LR/UR):} They are adaptations of the full range measure to account for the range below or above the expectation, respectively. They are formulated as:
      \begin{equation}\label{8}
       UR^{\mathrm{P}}(X)= \mathbb{E}^{\mathrm{P}}[X] -\inf X, \,\forall X\in L^{\infty};
       \end{equation}
       \begin{equation}\label{9}
       LR^{\mathrm{P}}(X)= \sup X - \mathbb{E}^{\mathrm{P}}[X], \,\forall X\in L^{\infty}.
      \end{equation}

  \item[$\bullet$] \textit{Variance (Var):} This is the most known deviation measure, being defined as:
   \begin{equation}\label{10}
    Var^{\mathrm{P}}(X)= \mathbb{E}^{\mathrm{P}}[(X-\mathbb{E}^{\mathrm{P}}[X])^{2}], \,\forall X\in L^{\infty}.
      \end{equation}

      It represents the second moment around expectation and has been considered as a proxy for risk in modern finance since the pioneering work of \cite{markowitz1952portfolio}, \cite{markowitz1970portfolio}.
  \item[$\bullet$] \textit{Standard Deviation (SD):} This variability measure is expressed as the root of the variance:
   \begin{equation}\label{11}
     SD^{\mathrm{P}}(X)=\left( \mathbb{E}^{\mathrm{P}}[(X-\mathbb{E}^{\mathrm{P}}[X])^{2}]\right)^{\frac{1}{2}}, \,\forall X\in L^{\infty}.
      \end{equation}

  \item[$\bullet$] \textit{Semi-Deviation ($SD_{-}/SD^{+}$):} The lower and upper semi-deviations are adaptations of the standard deviation that consider dispersion only from values, respectively, below or above the expectation in order to avoid symmetry. They are defined conform:
      \begin{equation}\label{12}
   SD_{\mp}^{\mathrm{P}}(X)=\left( \mathbb{E}^{\mathrm{P}}[((X-\mathbb{E}^{\mathrm{P}}[X])^{2})^{\mp}]\right)^{\frac{1}{2}}, \,\forall X\in L^{\infty}.
      \end{equation}

  \item[$\bullet$] \textit{Mean Gini coefficient (Gini):} This is a statistical coefficient that measures the variation degree in the set of values $X(\omega)$ as $\omega$ varies in $\Omega$, see \cite{shalit1984mean}, \cite{giorgi1993fresh}, \cite{giorgi2005bibliographic}, \cite{yitzhaki1998more}, \cite{ceriani2012origins}, \cite{furman2017gini} and \cite{berkhouch2018extended}. It is defined as:
      \begin{equation}\label{13}
    Gini^{\mathrm{P}}(X)= \mathbb{E}^{\mathrm{P}}[|X^{\ast}-X^{\ast\ast}|], \,\forall X\in L^{\infty};
      \end{equation}

      where, $X^{\ast}$ and $X^{\ast\ast}$ are two independent copies of $X$.
      Or, in terms of covariance ($Cov$):\[Gini^{\mathrm{P}}(X)= 4\,Cov^{\mathrm{P}}[X, F_{X,\mathrm{P}}(X)], \,\forall X\in L^{\infty}.\]
  \item[$\bullet$] \textit{Extended Gini coefficient (EGini):} This is a well-known coefficient in finance and economics characterized with a parameter of aversion degree $r\geq 1$, see \cite{yitzhaki1983extension}, \cite{yitzhaki2005properties}, \cite{yitzhaki2012gini}, \cite{mao2018risk} and \cite{berkhouch2018extended}. It is formulated conform:
      \begin{equation}\label{14}
       EGini^{\mathrm{P}}(X)= -2r\,Cov^{\mathrm{P}}[X,(1- F_{X,\mathrm{P}}(X))^{r-1}], \,\forall X\in L^{\infty}.
      \end{equation}

      For $r=2$, we get the Mean Gini coefficient.
\end{description}
\end{ex}

\section{Robust framework for spectral risk measures}
\label{sec3.3}

For a measurable space of scenarios $(S,\mathcal{S})$ and a weighting function $\phi$ (a priori specified), we consider a collection $\rho_{\phi}=(\rho_{\phi}^{s})_{s\in S}$ of spectral risk measures such that: $\rho_{\phi}^{s}: L^{\infty}(\Omega,\mathcal{F})\rightarrow \mathds{R}$, for every $s\in S$. We assume $\rho_{\phi,X}^{.}:= \rho_{\phi}^{.}(X): (S, \mathcal{S})\rightarrow \mathds{R}$ $\mathcal{S}$-measurable, for every financial position $X \in L^{\infty}$.

Throughout this paper, we adopt the following notations according to the context:
\[\rho_{\phi}(X):=(\rho_{\phi}^{s}(X))_{s\in S}=(\rho_{\phi,X}^{s})_{s\in S}:=\rho_{\phi,X}.\]

\begin{dfn}\label{d}
We call \textit{scenario-based spectral risk measurement} (abri. scenario-based SRM) every family of spectral risk measures with a weighting function $\phi$ such that:
\[\rho_{\phi}=\{\rho_{\phi}^{s}: L^{\infty}(\Omega, \mathcal{F})\rightarrow \mathds{R}, s\in S\}.\]
\end{dfn}

\begin{dfn}\label{e}
For a held financial position $X$, we define an \textit{uncertainty-free set} as:
\[\{\rho_{\phi,X} \in L^{\infty}(S,\mathcal{S}) \,s.t.: \rho_{\phi,X}^{s}=\rho_{\phi,X}^{s^{\prime}}, \forall s,s^{\prime}\in S\}.\]
\end{dfn}

We now define \textit{S-based risk measures} as introduced in \cite{wang2018scenario}:

\begin{dfn}\label{f}
For a scenario set $S$, a measure of risk $\rho$ is called S-based if: $\rho(X)=\rho(Y)$ for $X,Y\in L^{\infty}$ whenever $X\stackrel{d}{=_s}Y$.
\end{dfn}

Consider a scenario-based SRM, $\rho_{\phi}$, and let $\mu$ be a probability measure on $(S,\mathcal{S})$. We intend to capture the uncertain character of $\rho_{\phi}$, by composing $\rho_{\phi}$ with a risk measure on $L^{\infty}(S,\mathcal{S})$; the idea is, instead of picking a specific risk measure, we consider the whole set of candidates. Let $R$ be a risk measure on $L^{\infty}(S,\mathcal{S})$:

\begin{dfn}\label{g}
The \textit{composition} of $R$ and $\rho_{\phi}$ is a monetary risk measure defined as:
\begin{equation}\label{15}
  R\,o\,\rho_{\phi}(X)=R(-\rho_{\phi}(X)),\, \forall X\in L^{\infty}.
\end{equation}
\end{dfn}

By definition, the composition $R\,o\,\rho_{\phi}$ is $S$-based. This approach allows us to define new risk measures that account for uncertainty over our scenario-based SRM $\rho_{\phi}$. We list some relevant examples below:

\begin{ex}
\begin{description}
  \item[\,]
  \item[$\bullet$] \textit{Worst-case scenario ($\rho_{\phi}^{WC}$):} For $R=WC$, we get the worst-case scenario defined as:
  \begin{equation}\label{16}
    \rho_{\phi}^{WC}(X)= \displaystyle\sup_{s\in S}\rho_{\phi}^{s}(X),\,\forall X\in L^{\infty}.
  \end{equation}

  \item[$\bullet$] \textit{Scenario-based Expectation ($E_{\phi}$):} For $R=EL^{\mu}$, we introduce the scenario-based Expectation:
  \begin{equation}\label{17}
    E_{\phi}^{\mu}(X)= \int_{S}\rho_{\phi}^{s}(X)d\mu,\,\forall X\in L^{\infty}.
  \end{equation}

  \item[$\bullet$] \textit{Scenario-based Value-at-Risk ($VaR_{\alpha,\phi}$):} For $R=VaR_{\alpha}^{\mu}$, $\alpha\in [0,1]$, we define the scenario-based Value-at-Risk as:
      \begin{equation}\label{18}
      VaR_{\alpha,\phi}^{\mu}(X)=\inf \{x\in \mathds{R}:\mu(\rho_{\phi}(X)< x)\leq 1-\alpha\},\,\forall X\in L^{\infty}.
     \end{equation}

  \item[$\bullet$] \textit{Scenario-based entropic risk measure ($Entr_{\tau,\phi}$):} For $R=Entr_{\tau}^{\mu}$, $\tau >0$, we define the scenario-based entropic risk measure conform:
      \begin{equation}\label{19}
       Entr_{\tau,\phi}^{\mu}(X)= \displaystyle\sup_{\mathrm{Q}}\left\{E_{\phi}^{\mathrm{Q}}(X)- \frac{1}{\tau}H(\mathrm{Q}/\mu)\right\},\,\forall X\in L^{\infty}.
     \end{equation}

  \item[$\bullet$] \textit{Scenario-based Expected Shortfall ($ES_{\alpha,\phi}$):} For $R=ES_{\alpha}^{\mu}$, $\alpha\in [0,1)$, we introduce the scenario-based Expected Shortfall as:
      \begin{equation}\label{20}
       ES_{\alpha,\phi}^{\mu}(X)=\frac{1}{1-\alpha}\int_{\alpha}^{1}VaR_{\gamma,\phi}^{\mu}(X)d\gamma, \,\forall X\in L^{\infty}.
     \end{equation}

\end{description}
\end{ex}

\begin{rmk}
The above introduced risk measures are examples of interesting robust compositions for our scenario-based SRM $\rho_{\phi}$; nevertheless, $\rho_{\phi}$ could be superposed with any conceivable risk measure $R$ on $L^{\infty}(S,\mathcal{S})$. A discrete version of $E_{\phi}$, in (\ref{17}), may be introduced as \textit{scenario-based Weighted Average} ($WA_{\phi}$) across a set of $n$ scenarios $\{s_{1},...,s_{n}\}$:
\begin{equation}\label{21}
  WA_{\phi}^{\mu}(X)=\displaystyle\sum_{i=1}^{n} \mu(s_{i})\, \rho_{\phi}^{s_{i}}(X), \,\forall X\in L^{\infty}.
\end{equation}
This formulation could be more relevant for practical issues. Moreover, $\rho_{\phi}^{WC}$ and  $ES_{\alpha,\phi}$ are $S$-based and coherent; $Entr_{\tau,\phi}$ is $S$-based and convex; $WA_{\phi}$, $E_{\phi}$ and are $S$-based co-monotone and coherent.
\end{rmk}

\begin{rmk}
According to the properties of $\rho_{\phi}$ and $R$, and from Corollary 1 \cite{Righi2018theory} (for $f(\cdot)= R(- \cdot)$), we provide results regarding dual representation of the superposed robust risk measure $R\,o\,\rho_{\phi}$. Let $\rho_{\phi}$ be a scenario-based SRM and $R$ be a monetary risk measure. Then:
\begin{description}
  \item[$i)$] If $R$ is a law-invariant convex risk measure, then the representation is conform:
  \begin{equation}\label{41}
    R\,o\,\rho_{\phi}(X)= \displaystyle\sup_{\mu\in\mathcal{V}}\left\{\int_{[0,1)}ES_{\gamma}(X) dm^{\mu}- \delta_{R}(\mu)\right\},\, \forall X\in L^{\infty}.
  \end{equation}
  where, $m^{\mu}\in cl(\mathcal{M}_{E_{\phi}^{\mu}})$. For $\mathcal{V}$ the set of probability measures in $(S,\mathcal{S})$, and $\delta_{R}$ a penalty term; see \cite{Righi2018theory}.

  \item[$ii)$] If $R$ is a law-invariant coherent risk measure, then the representation becomes:
  \begin{equation}\label{41}
    R\,o\,\rho_{\phi}(X)= \displaystyle\sup_{\mu\in\mathcal{V_{R}}} \int_{[0,1)}ES_{\gamma}(X) dm^{\mu},\, \forall X\in L^{\infty}.
  \end{equation}
  where, $m^{\mu}\in cl(\mathcal{M}_{E_{\phi}^{\mu}})$ and $\mathcal{V_{R}}= \{\mu\in \mathcal{V}: R\,o\,\rho_{\phi}(X)\geq E_{\phi}^{\mu}(X), \,\forall X\in L^{\infty}\}$; see \cite{Righi2018theory}.

  \item[$iii)$] If $R$ is a law-invariant co-monotone coherent risk measure, then the representation becomes:
  \begin{equation}\label{41}
    R\,o\,\rho_{\phi}(X)= \int_{[0,1)}ES_{\gamma}(X) dm^{\mu},\, \forall X\in L^{\infty}.
  \end{equation}
  for some $m^{\mu}\in cl(\mathcal{M}_{E_{\phi}^{\mu}})$.

\end{description}
\end{rmk}

From an other perspective, Value-at-Risk ($VaR$) and Expected Shortfall ($ES$) are widely used in banking and insurance industries; therefore, we consider them as a building block for a second approach. First, by definition a spectral risk measure is expressed as weighted Value-at-Risk:
\begin{equation}\label{22}
  \rho_{\phi}^{s}(X)=\int_{0}^{1}VaR_{\gamma}^{s}(X)\,\phi(\gamma)d\gamma,\,\forall X\in L^{\infty}.
\end{equation}

Second, as a classic result from dual representation of law-invariant co-monotone coherent risk measures (see \cite{kusuoka2001law}, \cite{acerbi2002spectral} and  \cite{frittelli2005law}), a spectral risk measure may also be expressed as a mixture of Expected Shortfall (ES):
\begin{equation}\label{23}
\rho_{\phi}^{s}(X)=\int_{[0,1)} ES_{\gamma}^{s}(X)dm,\,\forall X\in L^{\infty}.
\end{equation}
for some probability measure $m$ on $[0,1)$.

Therefore, addressing the uncertainty issue over a scenario-based SRM $\rho_{\phi}=(\rho_{\phi}^{s})_{s\in S}$ can be turned to dealing with uncertainty over the collections $VaR_{\alpha}= (VaR_{\alpha}^{s})_{s\in S}$ and $ES_{\alpha}= (ES_{\alpha}^{s})_{s\in S}$. We then formulate the corresponding robust composition versions of $VaR_{\alpha}$ and $ES_{\alpha}$, for $R:L^{\infty}(S,\mathcal{S})\rightarrow \mathds{R}$:
\begin{equation}\label{24}
R\,o\,VaR_{\alpha}(X)= R(-VaR_{\alpha}(X)),\,\forall X\in L^{\infty}.
\end{equation}

\begin{equation}\label{25}
R\,o\,ES_{\alpha}(X)= R(-ES_{\alpha}(X)),\,\forall X\in L^{\infty}.
\end{equation}

Similarly, one may define an alternative version of $ES_{\alpha}$:
\begin{equation}\label{26}
ES_{\alpha,R}(X)= \frac{1}{1-\alpha}\int_{\alpha}^{1}R\,o\,VaR_{\gamma}(X)d\gamma, \,\forall X\in L^{\infty}.
\end{equation}

Based on this, we introduce below some new variants of robust risk measures that capture uncertainty over a scenario-based SRM $\rho_{\phi}$. Note that despite using the same notation, the below introduced risk measures are not necessarily equal.

\begin{dfn}
For a scenario-based SRM $\rho_{\phi}$ and a risk measure $R$ on $L^{\infty}(S,\mathcal{S})$, we introduce:
\begin{description}
  \item[$i)$]
  \begin{equation}\label{27}
   \rho_{\phi,R}(X)= \int_{0}^{1}R\,o\,VaR_{\gamma}(X)\, \phi(\gamma)d\gamma, \,\forall X\in L^{\infty}.
  \end{equation}
  \item[$ii)$]
  \begin{equation}\label{28}
   \rho_{\phi,R}(X)= \int_{[0,1)}R\,o\,ES_{\gamma}(X)dm, \,\forall X\in L^{\infty}.
  \end{equation}
  \item[$iii)$]
  \begin{equation}\label{29}
   \rho_{\phi,R}(X)= \int_{[0,1)}ES_{\gamma,R}(X)dm, \,\forall X\in L^{\infty}.
  \end{equation}
\end{description}
\end{dfn}

\begin{prp}\label{pr2}
Let $R$ be a risk measure on $L^{\infty}(S,\mathcal{S})$ and consider $\rho_{\phi,R}$ defined in (\ref{28}).
\item[$i)$] If $R$ is convex then $\rho_{\phi,R}$ is $S$-based and convex.
\item[$ii)$] If $R$ is coherent then $\rho_{\phi,R}$ is $S$-based and coherent.
\item[$iii)$] If $R$ is co-monotone then $\rho_{\phi,R}$ is $S$-based and co-monotone.
\end{prp}

\begin{proof}$\,$\\
  $i)$ $R$ is convex and $ES_{\gamma}$ is convex.\\
  We have $R\,o\,ES_{\gamma}$ is convex. Let $C \in \mathds{R}$:
  \begin{alignat*}{1}
   \rho_{\phi,R}(X+C)  &=  \int_{[0,1)}R\,o\,ES_{\gamma}(X+C)dm\\
     &= \int_{[0,1)}R\,o\,ES_{\gamma}(X)dm -C \\
     &= \rho_{\phi,R}(X)-C.\\
  \end{alignat*}

Let $X,Y \in L^{\infty}$ s.t.: $X\leq Y$, we have:
\begin{alignat*}{1}
   R\,o\,ES_{\gamma}(X) \geq R\,o\,ES_{\gamma}(Y)  &\Rightarrow  \int_{[0,1)}R\,o\,ES_{\gamma}(X)dm \geq \int_{[0,1)}R\,o\,ES_{\gamma}(Y)dm\\
     &\Rightarrow \rho_{\phi,R}(X)\geq \rho_{\phi,R}(Y). \\
  \end{alignat*}
Thus, $\rho_{\phi,R}$ is a monetary risk measure.\\
Let $\lambda \in [0,1]$ and $X,Y \in L^{\infty}$, since $R\,o\,ES_{\gamma}$ is convex we have:
\begin{alignat*}{1}
   & R\,o\,ES_{\gamma}(\lambda X+ (1-\lambda)Y) \leq \lambda\, R\,o\,ES_{\gamma}(X)+(1-\lambda)\, R\,o\,ES_{\gamma}(Y) \\
   &\Rightarrow  \int_{[0,1)}R\,o\,ES_{\gamma}(\lambda X+ (1-\lambda)Y)dm \leq \lambda \int_{[0,1)}R\,o\,ES_{\gamma}( X)dm + (1-\lambda) \int_{[0,1)}R\,o\,ES_{\gamma}(Y)dm \\
     &\Rightarrow \rho_{\phi,R}(\lambda X+ (1-\lambda)\,Y)\leq \lambda\, \rho_{\phi,R}(X)+(1-\lambda)\rho_{\phi,R}(Y). \\
  \end{alignat*}
Therefore, $\rho_{\phi,R}$ is convex.\\

$ii)$ $R$ is coherent.\\
$R\,o\,ES_{\gamma}$ is coherent. Let $\lambda\geq 0$, we get:
\begin{alignat*}{1}
   \rho_{\phi,R}(\lambda X)  &=  \int_{[0,1)}R\,o\,ES_{\gamma}(\lambda X)dm\\
     &= \int_{[0,1)}\lambda\, R\,o\,ES_{\gamma}(X)dm \\
     &= \lambda\, \rho_{\phi,R}(X).\\
  \end{alignat*}
Thus, $\rho_{\phi,R}$ is positive homogeneous. Then, $\rho_{\phi,R}$ is coherent. \\

$iii)$ $R$ is co-monotone.\\
$R\,o\,ES_{\gamma}$ is co-monotone. Let $X,Y$ be a co-monotone pair $\in L^{\infty}$, we have:
\begin{alignat*}{1}
   \rho_{\phi,R}( X+Y)  &=  \int_{[0,1)}R\,o\,ES_{\gamma}( X+Y)dm\\
     &= \int_{[0,1)} R\,o\,ES_{\gamma}(X)dm + \int_{[0,1)} R\,o\,ES_{\gamma}(Y)dm\\
     &= \rho_{\phi,R}(X)+\rho_{\phi,R}(Y).\\
  \end{alignat*}
Therefore, $\rho_{\phi,R}$ is co-monotone.
\end{proof}

\begin{rmk}
Spectral risk measures are a special case of the larger class of distortion risk measures. Therefore, the proposed approaches in this paper could be extended, in the same fashion, to the class of distortion measures of risk via the framework of Choquet integrals under a given capacity (see e.g., \cite{choquet1954theory}, \cite{schmeidler1986integral}, \cite{schmeidler1989subjective}, \cite{yaari1987dual} and \cite{wang2018characterization}).
\end{rmk}

\section{Measuring uncertainty}
As discussed early in the introduction and since measuring uncertainty means quantifying the variability over $\rho_{\phi}(X)=(\rho^{s}_{\phi}(X))_{s\in S}$ as long as $s$ varies across the scenario set $S$, we argue that this quantification should be done within the elements of $\rho_{\phi}(X)$ without reliance on a reference measure, as proposed in \cite{jokhadze2018measuring}. Thus, we claim that measuring uncertainty by a \textit{Deviation-based} approach is more relevant and suits much more the objective.

\begin{dfn}\label{j}
Consider a scenario-based SRM $\rho_{\phi}$, a measure of uncertainty $U$ is a non-negative mapping on $L_{\infty}(S,\mathcal{S})$ that evaluates uncertainty over $(\rho_{\phi}^{s}(X))_{s\in S}$ as $s$ varies in $S$, for a financial position $X$:
\begin{alignat*}{30}
   U: L_{\infty}(S,\mathcal{S})  &\rightarrow \mathds{R}_{+}\\
    \rho_{\phi}(X) &\mapsto U(\rho_{\phi}(X))\\
  \end{alignat*}
\end{dfn}

\begin{rmk}
Throughout this section we opt for the following notation: \[\rho_{\phi}(X):=\rho_{\phi,X}.\] The focus in this paper is especially on spectral risk measures; however, our approach could be extended to any other class of risk measures.
\end{rmk}
In this stated logic, we propose measuring uncertainty by using \textit{deviation measures} on $L^{\infty}(S,\mathcal{S})$: $U=D$. Moreover, the theoretical properties of a deviation measure (see Definition \ref{c}) are consistent and easily interpretable in the context of uncertainty measurement.  For instance, Non-Negativity assures that there is no variability in an \textit{uncertainty-free set} (see Definition \ref{e}); Translation Insensitivity indicates that uncertainty does not change if a certain constant value is added.

 Consider a scenario-based SRM, $\rho_{\phi}$, and a financial position $X$. Let $\mu$ be a probability measure on $(S,\mathcal{S})$ and $F_{\mu}$ be the corresponding distribution function for the random variable $\rho_{\phi,X}$ on $L^{\infty}(S,\mathcal{S})$. We provide some examples of uncertainty measures to illustrate our Deviation-based approach:
\[U(\rho_{\phi,X})= D(\rho_{\phi,X}):=D_{\phi}(X);\]
\begin{description}
  \item[$\bullet$] For $D=FR$, we introduce the \textit{Full Range-type uncertainty measure ($FR_{\phi}$):}
   \begin{equation}\label{31}
     FR_{\phi}(X)=\displaystyle\sup_{s\in S} \rho_{\phi,X}^{s} - \displaystyle\inf_{s\in S}\rho_{\phi,X}^{s}, \,\forall X\in L^{\infty}.
   \end{equation}

  This measure of uncertainty assesses the maximum distance between the candidates $\{\rho_{\phi,X}^{s}, s\in S\}$. Thus, it tends to overestimate uncertainty.

  \item[$\bullet$] For $D=LR/ D=UR$, we get a \textit{Lower/Upper Range-type uncertainty measure ($LR_{\phi}/UR_{\phi}$):}
  \begin{equation}\label{32}
     UR_{\phi}^{\mu}(X)= \mathbb{E}^{\mu}[\rho_{\phi,X}] -\displaystyle\inf_{s\in S}\rho_{\phi,X}^{s}, \,\forall X\in L^{\infty};
   \end{equation}

  \begin{equation}\label{33}
     LR_{\phi}^{\mu}(X)= \displaystyle\sup_{s\in S} \rho_{\phi,X}^{s} - \mathbb{E}^{\mu}[\rho_{\phi,X}], \,\forall X\in L^{\infty}.
   \end{equation}

  \item[$\bullet$] For $D=Var$, we define a \textit{Variance-type uncertainty measure ($Var_{\phi}$):}
  \begin{equation}\label{34}
     Var_{\phi}^{\mu}(X)= \mathbb{E}^{\mu}[(\rho_{\phi,X}-\mathbb{E}^{\mu}[\rho_{\phi,X}])^{2}], \,\forall X\in L^{\infty}.
   \end{equation}

  \item[$\bullet$] For $D=SD$, we define a \textit{Standard Deviation-type uncertainty measure ($SD_{\phi}$):}
  \begin{equation}\label{35}
     SD_{\phi}^{\mu}(X)= \left(\mathbb{E}^{\mu}[(\rho_{\phi,X}-\mathbb{E}^{\mu}[\rho_{\phi,X}])^{2}]\right)^{\frac{1}{2}}, \,\forall X\in L^{\infty}.
   \end{equation}

  \item[$\bullet$] For $D=Gini$, we introduce a \textit{Gini-type uncertainty measure ($Gini_{\phi}$):}
  \begin{equation}\label{36}
     Gini_{\phi}^{\mu}(X)= \mathbb{E}^{\mu}[|\rho_{\phi,X}^{\ast}-\rho_{\phi,X}^{\ast\ast}|], \,\forall X\in L^{\infty};
   \end{equation}

  where, $\rho_{\phi,X}^{\ast}$ and $\rho_{\phi,X}^{\ast\ast}$ are two independent copies of $\rho_{\phi,X}$. Or, in terms of covariance:
  \begin{equation}\label{37}
     Gini_{\phi}^{\mu}(X)= 4\,Cov^{\mu}[\rho_{\phi,X}, F_{\mu}(\rho_{\phi,X})], \,\forall X\in L^{\infty}.
   \end{equation}

  \item[$\bullet$] For $D=EGini$, we get an \textit{Extended Gini-type uncertainty measure ($EGini_{\phi}$):}
  \begin{equation}\label{38}
     EGini_{\phi}^{\mu}(X)= -2r\,Cov^{\mu}[\rho_{\phi,X},(1- F_{\mu}(\rho_{\phi,X}))^{r-1}], \,\forall X\in L^{\infty}.
   \end{equation}

  where, the parameter $r\geq1$ might be seen as a degree of uncertainty-aversion.
  This measure could be used in order to exhibit the attitude towards uncertainty of different profiles.
\end{description}

\section{A case study}

In this section we are going to provide an illustration using a case study. We consider a scenario-based SRM, $\rho_{\phi}$, where the spectral risk measure $\rho_{\phi}^{s}$ is the \textit{Extended Gini Shortfall} introduced in \cite{berkhouch2018extended} conform:
\begin{equation}\label{39}
  EGS_{r,p,\lambda}^{s}(X)= \int_{0}^{1} VaR^{s}_{\gamma}(X)\,\phi_{r,p}^{\lambda}(\gamma) d\gamma, \,\forall X\in L^{\infty}.
\end{equation}
where,\begin{equation}\label{40}
\phi_{r,p}^{\lambda}(\gamma)= \frac{1-p+2\lambda[(1-p)^{r-1}-r(1-\gamma)^{r-1}]}{(1-p)^{2}}\,\mathds{1}_{[p,1]}(\gamma),
\end{equation}
with $\gamma\in [0,1]$, $p\in(0,1)$, $r>1$, and $\lambda\in [0, 1/(2(r-1)(1-p)^{r-2})].$\\

Our data set consists of daily return $X$ of the NASDAQ index, covering the period from $Jun. 01, 2007$ to $Mar. 31, 2019$, with a total of $N=3080$ observations\footnote{\url{https://finance.yahoo.com}}. Figure \ref{Fig:Return Graph} shows the return evolution of the NASDAQ index over the observed period.
\begin{figure}[!]
  \centering
  \includegraphics[width=18cm,height=16cm]{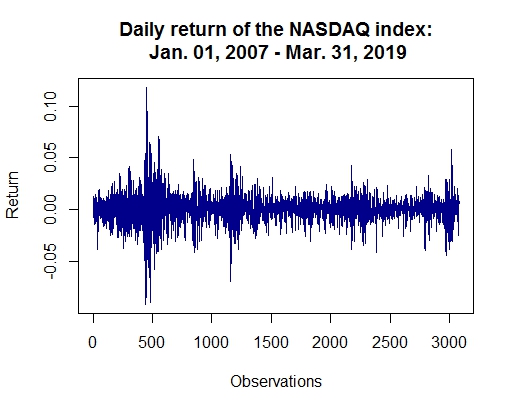}
\caption{Graph of the daily observed NASDAQ return.}\label{Fig:Return Graph}
\end{figure}

We consider sequences of three months from our data set. Thus, each trimester is regarded as a different economic situation (or scenario) that represents a different distribution of data. Therefore, we get a total of $n=49$ scenarios.

For each scenario $s$, we evaluate $EGS^{s}(X)$ for a confidence level $p=95\%$ and a risk-aversion degree $r=2$. To fulfill the coherency condition of $EGS$ we arbitrarily take $\lambda$ the midpoint of $[0, 1/(2(r-1)(1-p)^{r-2})]$, thus $\displaystyle \lambda=\frac{1}{4(r-1)(1-p)^{r-2}}$. Table \ref{Tab: EGS} reports the estimated values of $EGS$ according to each scenario\footnote{a: Jan., Feb., Mar./ b: Apr., May, Jun./ c: Jul., Aug., Sep./ d: Oct., Nov., Dec.}.
\begin{table}[!]
\begin{center}
\begin{tabular}{|c|c|c|}
  \hline
  & Scenarios & $EGS^{s}(X)$  \\\hline\hline
1 & 2007-a & 2.73\%\\
2 & 2007-b & 1.61\%\\
3 & 2007-c & 2.31\%\\
4 & 2007-d & 2.61\%\\
5 & 2008-a & 3.25\%\\
6 & 2008-b & 2.93\%\\
7 & 2008-c & 6.12\%\\
8 & 2008-d & 7.87\%\\
9 & 2009-a & 4.86\%\\
10 & 2009-b & 3.40\%\\
11 & 2009-c & 2.53\%\\
12 & 2009-d & 2.65\%\\
13 & 2010-a & 2.62\%\\
14 & 2010-b & 3.85\%\\
15 & 2010-c & 2.58\%\\
16 & 2010-d & 1.65\%\\
17 & 2011-a & 2.45\%\\
18 & 2011-b & 1.92\%\\
19 & 2011-c & 5.65\%\\
20 & 2011-d & 3.31\%\\
21 & 2012-a & 1.01\%\\
22 & 2012-b & 2.49\%\\
23 & 2012-c & 1.33\%\\
24 & 2012-d & 2.06\%\\
25 & 2013-a & 1.50\%\\
26 & 2013-b & 2.09\%\\
27 & 2013-c & 1.56\%\\
28 & 2013-d & 1.72\%\\
29 & 2014-a & 2.18\%\\
30 & 2014-b & 2.44\%\\
31 & 2014-c & 1.80\%\\
32 & 2014-d & 2.01\%\\
33 & 2015-a & 2.04\%\\
34 & 2015-b & 1.88\%\\
35 & 2015-c & 3.45\%\\
36 & 2015-d & 1.83\%\\
37 & 2016-a & 3.27\%\\
38 & 2016-b & 2.68\%\\
39 & 2016-c & 1.56\%\\
40 & 2016-d & 1.30\%\\
41 & 2017-a & 1.15\%\\
42 & 2017-b & 1.99\%\\
43 & 2017-c & 1.68\%\\
44 & 2017-d & 1.00\%\\
45 & 2018-a & 3.63\%\\
46 & 2018-b & 2.33\%\\
47 & 2018-c & 1.41\%\\
48 & 2018-d & 4.01\%\\
49 & 2019-a & 2.59\%\\
  \hline
\end{tabular}\caption{Estimated values of $EGS^{s}(X)$.}\label{Tab: EGS}
\end{center}
\end{table}

From Table \ref{Tab: EGS}, we notice discrepancy between the estimated values of $EGS$ under the different scenarios; for instance, we got extreme values of risk during the period between $2008$ and early $2009$ which is characterized by an extremely bad economic situation, due to the crisis. This emphasizes the point argued in this paper about uncertainty regarding the choice of a decent distribution. Moreover, it highlights the fact that this uncertainty could lead to concerning financial losses. At this point, one could measure uncertainty using the proposed Deviation-based approach, see Table \ref{Tab: U}.
\begin{table}[!]
\begin{center}
\begin{tabular}{c|c|c|c|c|c|c}
  \hline
  $U$ & $FR$ & $LR$ & $UR$ & $Var$ & $SD$ & $Gini$\\\hline\hline
  $U_{\phi_{r,p,\lambda}}(X)$ & 6.87\% & 5.28\% & 1.59\% & 0.02\% & 1.32\% & 1.31\%  \\
  \hline
\end{tabular}\caption{Uncertainty measurement.}\label{Tab: U}
\end{center}
\end{table}
 This could be used for comparison purposes between different scenario-based SRM.\\
 Now, in order to overcome the uncertainty issue we consider the entire set of scenarios, and we superpose our scenario-based SRM with risk measures such as Expectation, Value-at-Risk and Expected Shortfall. The results are shown in Table \ref{Tab: RoEGS}.\footnote{$A$: for Average, the equiprobable version of $WA$.}
\begin{table}[!]
\begin{center}
\begin{tabular}{c|c|c|c|c|c}
  \hline
  $R$ & $WC$ & $A$ & $WA$ & $VaR$ & $ES$\\\hline\hline
  $R\,o\,\rho_{\phi_{r,p,\lambda}}(X)$ & 7.87\%
 & 2.59\% & 3.26\% & 5.33\% & 6.55\%   \\
  \hline
\end{tabular}\caption{Estimated values for examples of robust risk measures.}\label{Tab: RoEGS}
\end{center}
\end{table}

From Table \ref{Tab: RoEGS}, we notice that while $WC$ actually overestimates risk, $A$ and $WA$ do underestimate risk. Moreover, between $VaR$ and $ES$, as enoughly argued in the literature, $ES$ seems more convenient.

This empirical exercise, based on the daily returns for the NASDAQ index between $Jan. 01, 2007$ and $Mar. 31, 2019$, is a historical approach that stresses the potential risk inherent to the uncertainty issue regarding the choice of distribution; furthermore, it illustrates the robust framework of risk measures proposed in the present paper.\\

\makeatletter

\makeatother

\bibliographystyle{apalike}
\bibliography{references}


\end{document}